\documentclass[conference]{IEEEtran} 
\usepackage{graphicx}
\usepackage{epsfig}
\usepackage{epstopdf}
\usepackage{amsmath,amsthm,amsfonts,amssymb}
\usepackage{multirow}
\usepackage{cite}
\usepackage{tikz}
\usepackage{ellipsis}
\usepackage{xcolor}
\usepackage{xifthen}
\usepackage{cancel}
\theoremstyle{definition} 
\newtheorem{prop}{Proposition}
\theoremstyle{definition} 
\theoremstyle{definition}

\pgfmathsetseed{1}
\usetikzlibrary{shapes,positioning,arrows,calc,graphs}
\usetikzlibrary{decorations.pathreplacing,decorations.markings,shapes.geometric}
\tikzset{naming/.style={align=center,font=\small}}
\tikzset{antenna/.style={insert path={-- coordinate (ant#1) ++(0,0.25) -- +(135:0.25) + (0,0) -- +(45:0.25)}}}
\tikzset{station/.style={naming,draw,shape=dart,shape border rotate=90, minimum width=10mm, minimum height=10mm,outer sep=0pt,inner sep=3pt}}
\tikzset{mobile/.style={naming,draw,shape=rectangle,minimum width=12mm,minimum height=6mm, outer sep=0pt,inner sep=3pt}}
\tikzset{radiation/.style={{decorate,decoration={expanding waves,angle=90,segment length=4pt}}}}

\newcommand{\MBS}[1]{%
\begin{tikzpicture}
\node[station] (base) {#1};

\draw[line join=bevel] (base.100) -- (base.80) -- (base.110) -- (base.70) -- (base.north west) -- (base.north east);
\draw[line join=bevel] (base.100) -- (base.70) (base.110) -- (base.north east);

\draw[line cap=rect] ([xshift=-.1768cm,yshift=.6pt]base.north -| base.right tail) [antenna=1];
\draw[line cap=rect] ([yshift=.6pt]ant1 |- base.north) -- node[above,shape=rectangle,inner ysep=+.3333em]{\dots} ([xshift=.1768cm,yshift=.6pt]base.north -| base.left tail) [antenna=2];

\end{tikzpicture}
}

\newcommand{\BS}[1]{%
\begin{tikzpicture}
\node[station] (base) {#1};

\draw[line join=bevel] (base.100) -- (base.80) -- (base.110) -- (base.70) -- (base.north west) -- (base.north east);
\draw[line join=bevel] (base.100) -- (base.70) (base.110) -- (base.north east);

\draw[line cap=rect] ([yshift=0pt]base.north) [antenna=1];
\end{tikzpicture}
}
\begin{document}

\title{On Edge Caching \\in the Presence of Malicious Users}

\author{\IEEEauthorblockN{Fr\'ed\'eric Gabry, Valerio Bioglio, Ingmar Land}
\IEEEauthorblockA{Mathematical and Algorithmic Sciences Lab\\ France Research Center, Huawei Technologies Co. Ltd.\\
Email: $\{$frederic.gabry,valerio.bioglio,ingmar.land$\}$@huawei.com}} 

\maketitle

\begin{abstract}
In this paper we investigate the problem of optimal cache placement in the presence of malicious mobile users in heterogeneous networks, where small-cell base stations are equipped with caches in order to reduce the overall backhaul load. In particular the malicious users aim at maximizing the congestion of files at the backhaul, i.e., at maximizing the average backhaul rate. For that adversarial model, we derive the achievable average backhaul rate of the heterogeneous network. Moreover, we study the system performance from a game-theoretic perspective, by naturally considering a zero-sum Stackelberg game between the macro-cell base station and the malicious users. We then thoroughly investigate the system performance in the presence of adversaries and we analyze the influence of the system parameters, such as the network topology and the capabilities of the small-cell base stations, on the overall performance of edge caching at the Stackelberg equilibrium. Our results highlight the impact of the malicious users on the overall caching performance of the network and they furthermore show the importance of an adversary-aware content placement at the small-cell base stations.

\end{abstract}

\section{Introduction}
\label{sec:intro}
A promising technique for future 5G wireless networks \cite{cache_magazine} consists of caching content at the wireless edge, as proposed in \cite{fem_2012}. The concept of edge caching emerges from the possibility of significantly reducing the backhaul traffic and thus the latency in content retrieval by bringing the content closer to the end users, e.g. mobile users. Building on the expected capabilities of future multi-tier networks \cite{hetnets}, also referred to as heterogeneous networks (HetNets), several recent works have investigated the potential benefits of caching data in densely deployed small-cell base stations (SBS) equipped with storage capabilities, as in, e.g., \cite{proactive}, called caching HetNets. 

In order to analyze the performance limits and trade-offs of caching in wireless networks, several perspectives have been envisaged. In particular, the information-theoretic perspective has gained considerable traction in recent years, for example in \cite{caching_networks}, where caching metrics are defined and analyzed for large networks. In \cite{fund_lim}, the authors study the fundamental performance limits of caching using network coding techniques. Other frameworks have been used to analyze caching networks, e.g. in \cite{modeling_tradeoffs} where the problem of cache content placement is studied in terms of outage probability.  Another particularly important performance measure for edge caching is the overall energy consumption, or equivalently the global energy efficiency of the network. Numerous works have investigated the performance of caching from an energy perspective: in \cite{EEICC} it was shown that caching at the edge provides significant gains in terms of energy efficiency while in \cite{ee_conext}, the authors study the tradeoff between transport and caching energy.

While the topic of wireless networks in the presence of adversaries has been deeply studied in the last decade, see e.g. \cite{sec09}, the interest in the security of caching HetNets against attackers has only grown recently. Moreover, most of the works which have investigated adversarial models in HetNets limited themselves to passive eavesdropping, i.e. the \emph{secrecy} metric was considered. For instance in \cite{hetnetssec}, the secrecy of HetNets without caching is studied while in \cite{limits_cache_sec} secrecy in caching networks is investigated from an information-theoretic perspective. In \cite{limits_d2d_sec} the secrecy of device-to-device networks is studied while in \cite{icc_sec} the authors introduced several eavesdropper models in caching HetNets. 

Departing from these works and building on the model introduced in \cite{globecom}, where the caching strategy was optimized with respect to backhaul load minimization without malicious users, we consider in this paper a novel model of adversaries aiming at maximizing the backhaul load. We thoroughly investigate this model in the paper and, namely, our main contributions are the following:

\begin{itemize}
\item We formally define the problem of caching at the wireless edge for HetNets with malicious mobile users.
\item We derive the achievable average backhaul rate and we investigate the system caching performance from a game-theoretic perspective. In particular we study the pure-strategy Stackelberg game between the macro-cell base station and the adversaries.
\item We thoroughly investigate the performance of the optimal secure scheme for a HetNet scenario of interest using numerical simulations. 
\item Our results highlight the considerable effect of the presence of malicious users on the performance of edge caching, as well as its impact on the optimal caching scheme in the Stackelberg game equilibrium.
\end{itemize}

This paper is organized as follows. 
In Section \ref{sec:model}, we define our system model, caching scheme and performance measures. 
In Section \ref{sec:perf}, we derive the achievable backhaul rate in the presence of adversaries. In Section \ref{sec:game}, we analyze the network from a game-theoretic perspective.
In Section \ref{sec:num} we thoroughly investigate the performance of our optimal schemes and we compare it to other caching schemes in a heterogeneous network scenario. 
Finally, Section~\ref{sec:conclusions} concludes this paper.

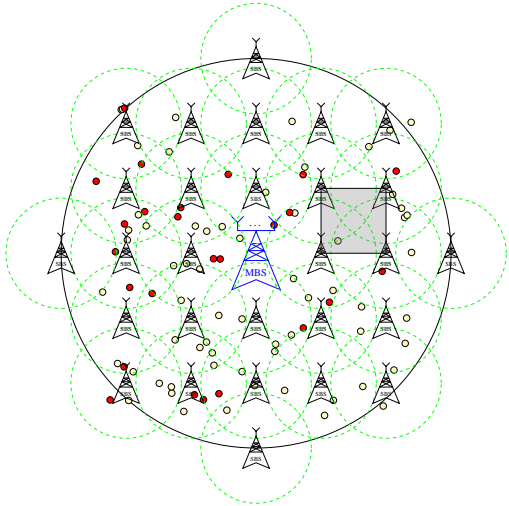
\begin{figure}[ht!]
 \centering
\resizebox{0.37\textwidth}{!}{  \begin{tikzpicture}
    \coordinate (Origin)   at (0,0);
    \coordinate (XAxisMin) at (-5,0);
    \coordinate (XAxisMax) at (5,0);
    \coordinate (YAxisMin) at (0,-5);
    \coordinate (YAxisMax) at (0,5);
    \draw [thick,black] (0,0) circle (6cm);
 
\foreach \x in {1,...,8}{
\foreach \y in {1,...,8}{
\pgfmathparse{rand} \pgfmathsetmacro\xa{\pgfmathresult}
\pgfmathparse{rand}\pgfmathsetmacro\ya{\pgfmathresult} 
\draw [thick,fill=yellow!30] (5*\xa,5*\ya) circle (0.1cm);
}}

\foreach \x in {1,...,5}{
\foreach \y in {1,...,5}{
\pgfmathparse{rand} \pgfmathsetmacro\xa{\pgfmathresult}
\pgfmathparse{rand}\pgfmathsetmacro\ya{\pgfmathresult} 
\draw [thick,fill=red] (5*\xa,5*\ya) circle (0.1cm);
}}

%    \draw[style=help lines,dashed] (-6,-6) grid[step=2cm] (6,6);
\node[blue,thick,inner sep=0pt] (MbS){\MBS{MBS}} (0,0);
%\draw[radiation] (MbS.north) -- +(90:0.5);
    \foreach \x in {-2,...,2}{
      \foreach \y in {-2,...,2}{
       \ifthenelse{\NOT 0 = \x \OR \NOT 0 = \y}{

        \node[scale=0.6] at (2*\x,2*\y) {\BS{SBS}} {};
        \draw[fill=none,dashed,green] (2*\x,2*\y) circle (1.7cm);
  \node  at (2*\x,2*\y) {};}{}
      }
    }   
    \filldraw[fill=gray, fill opacity=0.3, draw=black] (2,2)
        rectangle (4,0);
         \node[scale=0.6] at (0,6) {\BS{SBS}} {};
        \draw[fill=none,dashed,green] (0,6) circle (1.7cm);
  \node  at (0,6) {};
  
    \node[scale=0.6] at (0,-6) {\BS{SBS}} {};
        \draw[fill=none,dashed,green] (0,-6) circle (1.7cm);
  \node  at (0,-6) {};
  
    \node[scale=0.6] at (6,0) {\BS{SBS}} {};
        \draw[fill=none,dashed,green] (6,0) circle (1.7cm);
  \node  at (6,0) {};
  
    \node[scale=0.6] at (-6,0) {\BS{SBS}} {};
        \draw[fill=none,dashed,green] (-6,0) circle (1.7cm);
  \node  at (-6,0) {};
 \end{tikzpicture}}
  \caption{Heterogeneous network.}
  \label{fig:geo_topology}
\end{figure}

\section{System Model}
\label{sec:model}
In this section, we describe the heterogeneous network and the caching scheme investigated throughout the paper.
\subsection{Network Model}
\label{sub:Net}
In the following, we consider the network illustrated in Fig.~\ref{fig:geo_topology}, where
a macro-cell base station (MBS) serves the requests of $U$ wireless users. This network model is similar to the model in \cite{globecom}, with one fundamental difference.  Indeed we assume that there exist in the network $U_a = \alpha U$ adversary mobile users, with $\alpha \in [0,1]$, whose goal is to disrupt the system. 
In particular, their aim is to maximize the congestion of files at the backhaul, i.e., to maximize the backhaul rate, which will be formally defined in Section \ref{sec:perf}.  We should note that this adversarial model can also accurately represent a scenario where the MBS does not have correct information on the file popularity distribution, and hence is assuming the worst-case scenario in order to optimize file placement.

Users request files belonging to a library of $N$ files, $\mathcal{F} = \{ F_1, \dots, F_N \}$, each of size $B$ bits. Since files can be divided into blocks of the same size, the assumption of equally sized files is justifiable. 
We divide the wirless users in two sets, $\mathcal{U}_l$ and $\mathcal{U}_a$ which represent the set of legitimate and adversary users, respectively. 
The $(1-\alpha)U$ legitimate users in $\mathcal{U}_l$ request files according to a known file popularity distribution $p = \{ p_1, \dots, p_N \}$, where file $F_j$ is requested with probability $p_j$. 
On the other hand, the $U_a = \alpha U$ adversary users in $\mathcal{U}_a$ request files according a different strategy, which will be discussed in Section \ref{sec:game}. 
%In general, we can see an adversary as a user requesting files according to a distribution $\hat{p}$, with $\hat{p} \neq p$. 
%The files are requested according to a file popularity distribution  $p = \{ p_1, \dots, p_N \}$, where file $F_j$ is requested with probability $p_j$. 
%The MBS expects the files to be requested according to a file popularity distribution  $p = \{ p_1, \dots, p_N \}$, where file $F_j$ is expected to be requested with probability $p_j$.
In the figures, legitimate users are represented in yellow, while adversary users are drawn in red.  

\def\firstcircle{(1,1) circle (1.7cm)}
\def\secondcircle{(1,-1) circle (1.7cm)}
\def\thirdcircle{(-1,-1) circle (1.7cm)}
\def\fourthcircle{(-1,1) circle (1.7cm)}
\def\rectangle{(-1,1) rectangle (1,-1)}
\begin{figure}[ht!]
  \centering
  \resizebox{0.3\textwidth}{!}{
  \begin{tikzpicture}[scale=2]
  \clip (-1.5,-1.5) rectangle (1.5,1.5);
   \draw[black,fill=blue!20] (-1,1)
        rectangle (1,-1);

    \draw[fill=none,dashed,green] \firstcircle ;
    \draw[fill=none,dashed,green] \secondcircle;
    \draw[fill=none,dashed,green] \thirdcircle ;
 \draw[fill=none,dashed,green] \fourthcircle ;

    \begin{scope}
      \clip \rectangle;
      \fill[red!10] \firstcircle;
    \end{scope}
     \begin{scope}
      \clip \rectangle;
      \fill[red!10] \secondcircle;
    \end{scope}
     \begin{scope}
      \clip \rectangle;
      \fill[red!10] \thirdcircle;
    \end{scope}
    \begin{scope}
      \clip \rectangle;
      \fill[red!10] \fourthcircle;
    \end{scope}

    \begin{scope}
      \clip \rectangle;
      \clip \firstcircle;
      \fill[green!10] \secondcircle;
    \end{scope}

 \begin{scope}
      \clip \rectangle;
      \clip \firstcircle;
      \fill[green!10] \fourthcircle;
    \end{scope}
    
     \begin{scope}
      \clip \rectangle;
      \clip \thirdcircle;
      \fill[green!10] \secondcircle;
    \end{scope}
    
     \begin{scope}
      \clip \rectangle;
      \clip \thirdcircle;
      \fill[green!10] \fourthcircle;
    \end{scope}
    
     \begin{scope}
      \clip \rectangle;
      \clip \secondcircle;
      \clip \thirdcircle;
      \fill[yellow!10] \fourthcircle;
    \end{scope}
    
      \begin{scope}
      \clip \rectangle;
      \clip \secondcircle;
      \clip \thirdcircle;
      \fill[yellow!10] \firstcircle;
    \end{scope}
    
      \begin{scope}
      \clip \rectangle;
      \clip \firstcircle;
      \clip \fourthcircle;
      \fill[yellow!10] \secondcircle;
    \end{scope}
    
        \begin{scope}
      \clip \rectangle;
      \clip \firstcircle;
      \clip \fourthcircle;
      \fill[yellow!10] \thirdcircle;
    \end{scope}
    
          \begin{scope}
      \clip \rectangle;
      \clip \firstcircle;
      \clip \secondcircle;
      \clip \fourthcircle;
      \fill[blue!10] \thirdcircle;
    \end{scope}

\foreach \x in {1,...,3}{
\foreach \y in {1,...,3}{
\pgfmathparse{rand} \pgfmathsetmacro\xa{\pgfmathresult}
\pgfmathparse{rand}\pgfmathsetmacro\ya{\pgfmathresult} 
\draw [thick,fill=yellow!30] (\xa,\ya) circle (0.05cm);
}}

\foreach \x in {1,...,2}{
\foreach \y in {1,...,2}{
\pgfmathparse{rand} \pgfmathsetmacro\xa{\pgfmathresult}
\pgfmathparse{rand}\pgfmathsetmacro\ya{\pgfmathresult} 
\draw [thick,fill=red] (\xa,\ya) circle (0.05cm);
}}
 
   \node[scale=0.9] at (-1,-1) {\BS{SBS$_1$}} {};
       % \draw[fill=none,dashed,green] (-1,-1) circle (1.7cm);
        \node[scale=0.9] at (-1,1) {\BS{SBS$_2$}} {};
      %  \draw[fill=none,dashed,green] (-1,1) circle (1.7cm);
        \node[scale=0.9] at (1,-1) {\BS{SBS$_3$}} {};
      %  \draw[fill=none,dashed,green] (1,-1) circle (1.7cm);
        \node[scale=0.9] at (1,1) {\BS{SBS$_4$}} {};
       % \draw[fill=none,dashed,green] (1,1) circle (1.7cm);
         \draw[fill=none,dashed,green] \firstcircle ;
    \draw[fill=none,dashed,green] \secondcircle;
    \draw[fill=none,dashed,green] \thirdcircle ;
 \draw[fill=none,dashed,green] \fourthcircle ;
   \end{tikzpicture}}
  \caption{Small cell topology.}
  \label{fig:square}
\end{figure}
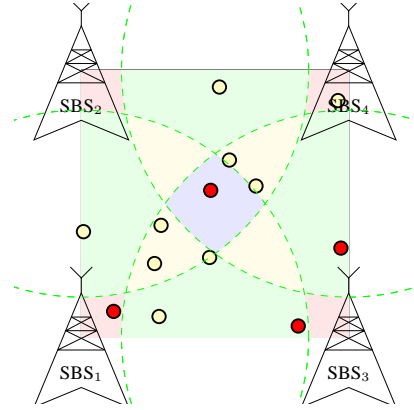

In order to serve user requests within short distance, $N_{\text{SBS}}$ small-cell base stations are deployed in the coverage range of the MBS. 
Each SBS is equipped with a cache of size $M \cdot B$ bits (i.e., it can store up to $M < N$ complete files). 
We assume that content is delivered without errors as long as a user is within the coverage range of the base station. 
Each user $u \in \{ 1, \dots, U \}$ requests for a file in $\mathcal{F}$, and each request is initially served by the $d_u$ $\text{SBS}$s within the coverage range with the cache content. 
In the following, we call $\gamma_i$ the probability for a user to be served by $d_u = i$ SBSs, which depends on its location and on the deployment of the SBSs in the area (see Fig.~\ref{fig:square}). 
If the requested file is not completely present in the SBSs' caches, the missing data has to be recovered by the MBS and sent to the user via the backhaul link, thus increasing the backhaul traffic. The purpose of caching files at the SBSs is then to minimize this backhaul rate, i.e., the number of bits that the MBS has to send directly to the users. % To that aim, the MBS exploits its knowledge of the expected file popularity and  the topology of the network, as illustrated in Fig.~\ref{fig:graph}, to fill the caches of the SBSs. 
%In fact, the topology of the connections may vary during time, and at each instant $t$, the connection network has a structure . 

\subsection{Coded Caching Scheme}
\label{sub:scheme}
A caching scheme can be divided into two phases, namely the \emph{placement phase} and the \emph{delivery phase}. 
In the placement phase, the $\text{SBS}$s caches are filled according to the chosen placement strategy. 
This phase typically occurs at a moment with a low amount of network traffic, e.g. at night. 
In the delivery phase, the users send requests for files in the library, which are initially served by the $\text{SBS}$s covering their locations. 
If the SBSs cannot send enough information for the users to recover the file, the $\text{MBS}$ has to retrieve the missing information through the backhaul and deliver it to the users. 

In \cite{globecom}, the authors showed that storing MDS encoded packets is always better than storing fragments in terms of retrieval probability. 
For MDS($k,n$) codes, such as Reed-Solomon (RS) codes, $k \geq n$ encoded packets are created such that any subset of $n$ packets are necessary and sufficient to recover the initial information. As a consequence, in the following we will study the performance of a MDS coded caching scheme. 

Each file of the library is initially split into $n$ \emph{fragments}, i.e., $F_j = \{ f_{1}^{(j)}, \cdots, f_{n}^{(j)} \}$ for all $1 \leq j \leq N$. The caching scheme is then defined as follows: 
\begin{enumerate}
\item \emph{Placement phase}: The $\text{MBS}$ uses the $n$ fragments of each file to create $k_j = n + (N_{SBS} - 1) m_j$ encoded packets using a MDS code, sending $m_j$ of them to each SBS. 
We note that the MBS keeps $n-m_j$ encoded packets unsent. 
Each $\text{SBS}$ receives $m_j$ different encoded packets for each file $F_j$ to be stored in its cache, with $\mathbf{m}=[m_{1} \cdots m_{N}]$ and $\sum_{j}m_j \leq M$. 
\item \emph{Delivery phase}: A user, requesting for file $F_j$, and within coverage range of $d_u \geq 1$ $\text{SBS}$s, receives $m_j d_u$ different encoded packets. 
If $m_j d_u \geq n$, the file can be recovered by the user due to the MDS-property of codes. 
Otherwise, the $\text{MBS}$ has to send the remaining $n - m_j d_u$ encoded packets from its collection of encoded packets. % randomly drawn from its collection. 
Due to the MDS-property of the codes, the user does not receive replicated packets, and it can hence recover the requested file. 
\end{enumerate}

%\input{figures/graph.tex}

%\subsection{Average Backhaul Rate Calculation}

\section{Average Backhaul Rate}
\label{sec:perf}
In this section, we derive the average backhaul rate per user in the heterogeneous network with malicious users. 
In the following, we denote by $\mathbf{q}$ the proportional placement of files stored in the SBSs caches, i.e., $\mathbf{q} \triangleq \mathbf{m}/n$. Hence $q_j = m_j/n$ represents the proportion of file $F_j$ stored in the caches.

If we denote by $R_l$ the average rate of a legitimate user $u_l \in \mathcal{U}_l$ and $R_a$ is the average rate of an adversary user $u_a \in \mathcal{U}_a$, the total average backhaul rate of the system is given by
\begin{equation}
\label{eq:R_tot}
\bar{R} = \alpha R_a + (1-\alpha) R_l.
\end{equation}
\subsection{Derivation of  $R_l$}
The average backhaul rate of a user requesting files according to a distribution $p$ can be calculated according to the following Proposition in \cite{globecom}: 
\begin{prop}[\cite{globecom} Prop. 1]
\label{prop:rate}
The average backhaul rate for an encoded caching placement scheme $\mathcal{C}_{\mathbf{m}}^{\text{MDS}}$ defined by the placement $\mathbf{q}=[q_{1} \cdots q_{N}]$ with expected requests vector $p$ can be calculated as 
\begin{equation}
\label{eq:th_rate}
R_{(\mathcal{C}_{\mathbf{q}}^{\text{MDS}})} = \sum_{d=1}^{S} \sum_{j=1}^{N} \gamma_d p_j  \max \left( 1 - d q_j, 0 \right),
\end{equation}
where $S \leq N_{\text{SBS}}$ is the maximum number of $\text{SBS}$s serving a user.
Hence, the average backhaul rate of a legitimate user can be  calculated as
\begin{equation}
\label{eq:R_l}
R_l = \sum_{d=1}^{S} \sum_{j=1}^{N} \gamma_d p_j \max \left( 1 - d q_j, 0 \right).
\end{equation}
\end{prop}

\subsection{Derivation of  $R_a$}
Similarly, we can derive the average backhaul rate of an adversary user $u_a \in \mathcal{U}_a$ as
\begin{equation}
\label{eq:R_a}
R_a = \sum_{d=1}^{S} \sum_{j=1}^{N} \gamma_d \hat{p_j} \max \left( 1 - d q_j, 0 \right),
\end{equation}
where $\hat{p_j}$ is the popularity of file $j$ induced by the requests of the adversaries. This result can be described in a different manner by writing the average backhaul rate of an adversarial user as 
\begin{equation*}
%\label{eq:R_a}
R_a = \frac{1}{|\mathcal{U}_a|} \sum_{u=1}^{|\mathcal{U}_a|}\sum_{d=1}^{S}  \gamma_d \max \left( 1 - d q_{j_u}, 0 \right),
\end{equation*}
where $j_u$ is the file requested by malicious user $u$. This description is equivalent to the one of Equation \eqref{eq:R_a}.

\section{Game-Theoretic Analysis}
\label{sec:game}
In this section we analyze the competitive interaction between the  $U_a$ adversaries and the MBS using a game-theoretic framework, which is especially suitable for our scenario.

\subsection{Definition of the Game}
Formally, we define the Stackelberg game $\mathcal{G}$ between the MBS and the malicious users as follows:
\begin{itemize}
\item \textbf{Leader:} the MBS, with utility function $\mathcal{U}_{\text{MBS}} = -\bar{R}$ and strategy space given by the placement scheme $\mathbf{q} = [q_1, \dots, q_N]$;
\item \textbf{Follower:} the $U_a$ malicious users, with utility function $\mathcal{U}_{\text{adv}} = -\mathcal{U}_{\text{MBS}}=\bar{R}$ and strategy space $\mathbf{j}= [j_1, \dots, j_{U_a}]$, i.e., the files requested by each malicious user. 
\end{itemize}

\subsection{Solution of the Game}
In this section, we solve the Stackelberg game defined previously. Note that we restrict the strategy space of players to pure strategies, i.e. we do not consider mixed strategies for the MBS; this restriction is justifiable from practical system considerations.

\begin{prop}
The Stackelberg game $\mathcal{G}$ has a unique solution $(\bf{q}^*, \bf{j}^*)$, leading to a utility $\bar{R}(\bf{q}^*, \bf{j}^*)$.
\end{prop}

\begin{proof}
Each player aims to maximize its own utility. In particular, the $U_a$ malicious users, acting as the follower of the Stackelberg game $\mathcal{G}$, solve first the following optimization:
\begin{equation}
\max_{\bf{j}} \bar{R},
\end{equation}
for $\bf{q}$ fixed. We can easily derive that $\forall \bf{q}$ and $\forall u \in \mathcal{U}_a$, the optimal strategy of $u$ is given by 
\begin{equation}
j^{*}_u(\mathbf{q})=\arg\min_i q(i),
\end{equation}
i.e., the optimal strategy of every attacker is to request for the file $j^{*}$ with minimal $q(j^{*})$. This result can also be readily obtained by deriving $\frac{\partial \bar{R}}{\partial{q_{j_u}}}$.
Hence, the request distribution of adversary users is given by
\begin{equation*}
\hat{p^{*}_j} = \left\{ \begin{array}{l c}
1 & \text{if } j = j^{*} \\
0 & \text{otherwise}.
\end{array} \right.
\end{equation*}
As a consequence, the expected backhaul rate for an adversary user can be simplified as
\begin{align}
%\label{eq:R_a}
R_a &= \sum_{d=1}^{S} \gamma_d \max \left( 1 - d q_{j^{*}}, 0 \right) \\
&= \sum_{d=1}^{S} \gamma_d \max \left( 1 - d \min(\mathbf{q}), 0 \right).
\end{align}
Now the MBS as the leader of the Stackelberg game solves
\begin{align}
\bar{R}(\bf{q}^*, \bf{j}^*)&= \min_{\mathbf{q}}\bar{R}(\bf{q}, \bf{j}^*)\\
&=\min_{\mathbf{q}} (1-\alpha)\sum_{d=1}^{S} \sum_{j=1}^{N} \gamma_d p_j \max \left( 1 - d q_j, 0 \right) \nonumber \\&\:\:\:\:+\alpha \sum_{d=1}^{S} \gamma_d \max \left( 1 - d \min(\mathbf{q}), 0 \right). \label{eq:SE}
\end{align}
Equation \eqref{eq:SE} is convex in $\bf{q}$ and hence the minimization problem $\min_{\mathbf{q}}\bar{R}(\bf{q}, \bf{j}^*)$ admits a solution $\bf{q}^*$ which concludes the proof of the proposition.
\end{proof}

The extreme cases $\alpha \rightarrow 0$ and $\alpha \rightarrow 1$ are of particular interest and are discussed in the next section.

\subsection{Extreme Cases}
\label{sec:extreme}
\paragraph{No malicious user, $\alpha \rightarrow 0$}
If $\alpha \rightarrow 0$, according to (\ref{eq:R_tot}) we have that $\bar{R} \rightarrow R_l$. 
In this case, the number of adversaries is negligible compared to the number of legitimate users, and the objective function of the optimization problem turns out to be the same of \cite{globecom}. 

\paragraph{No legitimate user, $\alpha \rightarrow 1$}
If $\alpha \rightarrow 1$, the Stackelberg equilibrium strategy of the MBS, given by Equation \eqref{eq:SE} becomes:
\begin{equation}
\min_{\mathbf{q}} \sum_{d=1}^{S} \gamma_d \max \left( 1 - d \min(\mathbf{q}, 0 \right)).
\label{eq:WCS}
\end{equation}
This represents the worst-case scenario, where all the users are adversaries. 
In this case, the MBS has to solve the optimization problem \eqref{eq:WCS} which is equivalent after some manipulations to
\begin{equation*}
\max_{\mathbf{q}} \left( \min_{j} (q_j) \right),
\end{equation*}
for which the Stackelberg equilibrium strategy is given by the uniform solution, i.e., 
\begin{equation}
\mathbf{q}^{*} = \left[\frac{M}{N} \ldots \frac{M}{N}\right]. 
\end{equation}
The Stackelberg equilibrium results in that case in the utility:
\begin{equation}
\bar{R}(\mathbf{q}^*, \mathbf{j}^*) = \sum_{d=1}^{S} \gamma_d \max \left( 1 - d\frac{M}{N}, 0 \right).
\end{equation}

\section{Numerical Illustrations}
\label{sec:num}
In this section we investigate the performance of the heterogeneous network in terms of backhaul rate of our coded caching scheme in the presence of adversary users. 
We evaluate the performance in a HetNet topology of particular interest, and we should note that our numerical results can be further generalized to any network topology. 
%In particular the challenging problem of the optimization of the locations of the small-cell base stations is of considerable practical and theoretical interest, but outside the scope of this paper. 

\subsection{Network Topology}
\label{num:topology}
We consider the heterogeneous network depicted in Fig.~\ref{fig:geo_topology}. 
In this scenario, the MBS has a coverage area of radius of $D=500$ meters. 
The SBSs are deployed in a regular grid, with a distance $d=60$ meters between SBSs. 
Each SBS has a coverage area of radius $r$, with $ d/\sqrt{2} \leq r \leq d$. 
This means that the coverage areas of the SBSs are overlapping, as shown in Fig.~\ref{fig:square}, which corresponds to $S=4$. 

Considering a uniform density $\rho$ of the users, the value of $\gamma_d$ is calculated from the sum of the coverage areas. 
In practice, if we call $\mathcal{A}_d$ the sum of the areas where a user can be served by $d$ SBSs, the probability $\gamma_d$ that a user is served by $d$ SBSs is given by
\begin{equation*}
\gamma_d = \frac{ \mathcal{A}_d}{\sum_{i=1}^{S} \mathcal{A}_i}.
\end{equation*}

We consider a uniform distribution of the users, with density $\rho=0.05$ users$/m^{2}$. Using these values corresponds to $316$ small-cell base stations being deployed, covering $U=39,269$ mobile users. 
The request probability of the files $p=[p_{1} \cdots p_{N}]$ for the legitimate users $u \in \mathcal{U}_l$ is distributed according to a Zipf law of parameter $z$, i.e.,
\begin{equation*}
p_{j}=\frac{1/j^{z}}{\sum_j 1/j^{z}},
\end{equation*}
where $z$ represents the skewness of the distribution  \cite{zipf}.

\subsection{Numerical Results}

In the following, we study the performance of the coded cache placement in the presence of adversary users $u \in \mathcal{U}_a$. 
%Unless specified otherwise, the SBSs have a coverage area of $r=60$ meters with a storage capacity of $M=20$ files. 
The library contains $N=200$ files, whose popularity for legitimate users follows a Zipf distribution of parameter $z=0.7$. 

\begin{figure}[!t]
\centering
\includegraphics[width=0.44\textwidth]{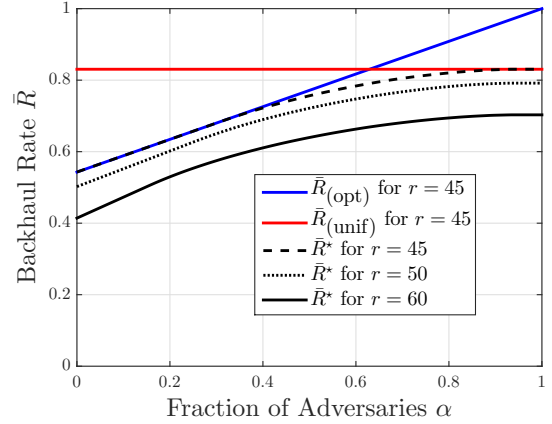}
\caption{Average backhaul rate as a function of the adversaries proportion $\alpha$, with N$=200$ files, $z=0.7$, $M=20$ files and varying $r$.}
\label{fig:R_alpha}
\end{figure}
In Figure \ref{fig:R_alpha}, we show the average backhaul rate $\bar{R}(\mathbf{q}^{*},\mathbf{j}^{*})$ in the Stackelberg equilibrium as a function of the fraction of adversaries $\alpha$. 
Each curve is calculated for a different value of the coverage range of the SBSs $r$. 
As a reference, we also show the backhaul rate of the uniform and the optimal caching placement presented in \cite{globecom}, i.e., the Stackelberg equilibrium strategies for $\alpha=1$ and $\alpha=0$, respectively.  
%As expected, the number of adversaries has an impact on the backhaul rate. 
As expected, the average backhaul rate grows with the number of adversaries. 
It is interesting to note that $r$ has a linear impact on the backhaul rate, where the rate approximately decreases by $0.04$ per meter of coverage. 
Moreover, by comparing the curves with the optimal caching strategies, we see that each curve is initially similar to the optimal caching strategy without adversaries, and finally converges to the uniform strategy. This highlights the existence of three regimes depending on $\alpha$:
\begin{itemize}
\item{$0 < \alpha < \alpha^{(1)}_{\text{thr}}$:} The Stackelberg equilibrium strategy of the MBS is independent of the malicious users' attacks, i.e. the MBS minimizes $R_l$.
\item{$\alpha^{(1)}_{\text{thr}} < \alpha < \alpha^{(2)}_{\text{thr}}$:} The MBS should adaptively optimize its placement depending on $\alpha$.
\item{$\alpha^{(2)}_{\text{thr}} < \alpha < 1$:} The MBS is forced to choose the minimax solution, i.e., the uniform placement.
\end{itemize}

\begin{figure}[!t]
\centering
\includegraphics[width=0.44\textwidth]{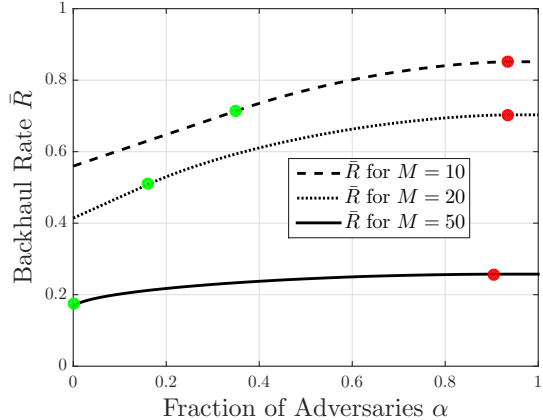}
\caption{Average backhaul rate as a function of the adversaries proportion $\alpha$, with N$=200$ files, $z=0.7$, $r=60$ meters and varying $M$.}
\label{fig:R_alpha_M}
\end{figure}
Figure \ref{fig:R_alpha_M} depicts the average backhaul rate as a function of the fraction of adversaries $\alpha$ with varying parameter the cache size $M$. 
As before, the number of adversaries has an impact on the backhaul rate. 
However, the larger is the size $M$ of the cache, the less impactful the attacks are. 
This is due to the fact that the optimal caching strategy in case of $\alpha=1$ is to store fragments uniformly, which is far away from the optimal caching placement if the cache size is small \cite{globecom}. 
Moreover, the points where the curve branches off from the optimal caching and gathers the uniform strategy seems to depend on the cache size $M$.   
To study in detail the behavior of the branching and gathering point, we depicted them as green and red circles, corresponding to $\alpha^{(1)}_{\text{thr}}$ and $\alpha^{(2)}_{\text{thr}}$, respectively.
Surprisingly, the larger the cache size, the smaller $\alpha^{(1)}$ and $\alpha^{(2)}$ are. 
This is due to the fact that, as shown in Figure \ref{fig:R_alpha}, the penalty of using a uniform placement is not so important in case of large cache size. 
In this case, the large cache size provides the possibility to timely react to the presence of adversaries, slightly changing the values of the number of fragments stored.  
On the other hand, if the cache size is small, the presence of a small number of adversaries does not have a large impact on the caching system, since the uniform caching performs far worse than the optimal strategy. 
In this case, it is better to keep the optimal strategy even in presence of a small number of adversaries, since the small size of the cache does not offer the possibility to substitute many packets.

\begin{figure}[!t]
\centering
\includegraphics[width=0.44\textwidth]{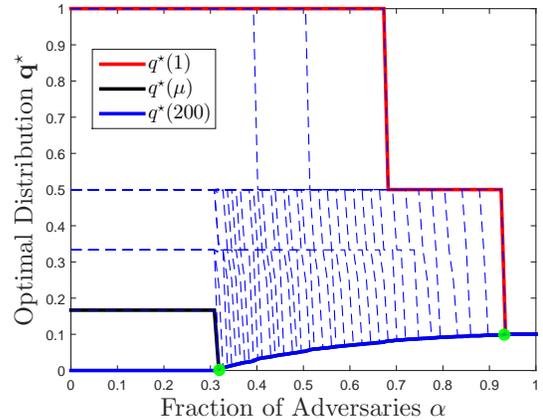}
\caption{Optimal placement $\mathbf{q}^{*}$ as a function of $\alpha$, with N$=200$ files, $z=0.7$, $M=20$ and $r=45$ meters.}
\label{fig:Q}
\end{figure}
To study in detail the behavior of the branching and gathering point, in Figure \ref{fig:Q} we show the entries of $\mathbf{q}^{*}$ as a function of $\alpha$.
In the Figure, we highlight the values of the minimum entry $q^*_{1}$, the maximum entry $q^*_{200}$ and the entry $q^*_{\mu}$, with $\mu \triangleq \max_\mu (q(\mu) \neq 0)$. 
For $\alpha=0$, we have that $q^*_{200}=0$, since the popularity of file $F_{200}$ is too small to be stored in such a small cache. 
For a similar reason, $q^*_{1}=1$, since the file is so popular that it has to be stored in every cache. 
The other files have a different amount of fragments stored. We should note that for  $\alpha=0$, the solution of the optimization problem of \cite{globecom} has a very regular structure, as we observe that $\mathbf{q}^{*}=[1,\frac{1}{2}, \ldots \frac{1}{S} \ldots 0]$.
Around $\alpha^{(1)}_{\text{thr}}=0.32$, the penalty due to the adversaries starts to have an impact on the performance of the system, hence the less popular files begin to be stored. 
Since all the adversaries ask for the less stored file, all these low popular files tend to have the same number of fragment stored. 
As the ratio of adversaries $\alpha$ grows, more files lose fragments in favor of less popular ones.  
Around $\alpha^{(2)}_{\text{thr}}=0.93$, all the files have the same number of fragments stored, and the uniform caching strategy is reached, i.e., the extreme $\alpha \rightarrow 1$ scenario described in Section \ref{sec:extreme}.

%\vspace{-0.5cm}
\section{Conclusions}
\label{sec:conclusions}
We investigated the problem of MDS-encoded content placement at the cache-equipped small-cell base stations at the wireless edge in the presence of malicious end users. First, we derived the achievable average backhaul rate for our scenario with adversaries. We then investigated the competitive interaction between the macro-cell base station and the malicious users from a Stackelberg game perspective, for which we derived the Stackelberg equilibrium in the general case and in the worst-case scenario. We then thoroughly analyzed the performance of caching at the wireless edge in the presence of malicious users for a relevant heterogeneous scenario by comparing it to the optimal placement without adversaries and to an uniform placement, which was shown to be optimal when the number of adversaries grows large. We also studied the influence of the key parameters, such as the capabilities of the small-cell base stations and the network topology. Our findings showed the negative impact of the presence of attackers on the caching performance and also showed the existence of three regimes depending on the fraction of adversaries $\alpha$. Hence our results showed the crucial importance of optimizing adaptively content placement depending on the number of malicious users. Further work will include the estimation of the number of attackers depending on the request distribution, and the use of a different game-theoretic framework to account for the uncertainty on the players' strategies, e.g. Bayesian game models.
%\vspace{+34ex}
%\input{appendix.tex}
%\vspace{3cm}

\bibliographystyle{IEEEbib}
\bibliography{distributed_caching}

\end{document}